\documentclass[12pt,a4paper]{article}

\usepackage{amsmath,amssymb,amsthm}
\usepackage{braket}
\usepackage{graphicx}
\usepackage{geometry}
\usepackage{float}

\newtheorem{theorem}{Theorem}[section]
\newtheorem{corollary}[theorem]{Corollary}
\newtheorem{proposition}[theorem]{Proposition}

\newtheorem{remark}[theorem]{Remark}

\title{Spectral-Gap Bounds and Timescales for Purity Loss in Hamiltonian--Pointer Interactions}

\author{
Orhan Amirov$^{1}$ and Necati \c{C}elik$^{2}$\\[1ex]
\small $^{1}$Department of Mathematics, Faculty of Science,\\
\small Atat\"urk University, Erzurum, T\"urkiye\\
\small $^{2}$Department of Physics Engineering, Faculty of Engineering and Natural Sciences,\\
\small G\"um\"u\c{s}hane University, G\"um\"u\c{s}hane, T\"urkiye
}

\date{}

\begin{document}

\maketitle

\begin{abstract}
We investigate the purity dynamics of a quantum system coupled to a
continuous-variable pointer through a von Neumann-type interaction
Hamiltonian of the form
$\hat H_{\mathrm{int}}=g\,\hat H\otimes\hat p$.
For an initially Gaussian pointer state, the interaction generates
energy-dependent conditional translations whose mutual overlaps are
determined explicitly by the populated spectral separations of the
system Hamiltonian. After tracing out the pointer degrees of freedom,
we obtain the reduced density operator and derive an exact analytical
expression for the time-dependent purity.

Using this expression, we establish two-sided purity bounds governed by
the minimum and maximum nonzero energy gaps on the populated spectral
support. These bounds provide a state-dependent spectral
characterization of the loss of purity and become exact for two-level
systems. We further show that the short-time decrease of purity is
controlled by the Hamiltonian variance of the initial state, with
$\mathcal P''(0)=-(g^2/\sigma^2)
\operatorname{Var}_{\psi_S}(\hat H)$.
In addition, an explicit sufficient timescale is derived for the purity
to approach its asymptotic value within a prescribed tolerance,
revealing the scaling
$t_{\varepsilon}\propto\sigma/(|g|\Delta_{\min})$.
Finally, the general results are illustrated for an equally weighted
$N$-level system with an equally spaced spectrum, for which the
asymptotic purity is $1/N$. The analysis clarifies the distinct roles
of spectral separation, energy variance, coupling strength, and pointer
width in Hamiltonian-conditioned purity loss.
\end{abstract}

\noindent\textbf{Keywords:}
quantum purity; von Neumann measurement; continuous-variable pointer;
spectral gap; decoherence; Hamiltonian variance; Gaussian pointer;
quantum measurement dynamics.

\section{Introduction}
\label{sec:introduction}

The interaction between a quantum system and a measurement apparatus
provides one of the fundamental mechanisms through which information
about a quantum observable is transferred to an experimentally
accessible degree of freedom. In the von Neumann measurement
framework, this transfer is modeled by coupling a system observable to
the momentum of a continuous-variable pointer, so that different
eigenvalues of the observable generate different translations of the
pointer wave packet \cite{vonNeumann2018}. The resulting conditional
pointer states provide a direct dynamical connection between the
spectral properties of the measured observable and the distinguishability
of the corresponding measurement outcomes.

The loss of coherence associated with such system--apparatus
correlations is closely related to the broader theory of decoherence
and environment-induced superselection. Early studies established how
interactions with external degrees of freedom suppress coherences and
select preferred states or pointer bases
\cite{ccelik2025quantum,Zurek1981,Zurek1982,JoosZeh1985}. These ideas were subsequently
developed in the context of coherent states, quantum measurement, and
the emergence of classical behavior
\cite{ZurekHabibPaz1993,Venugopalan2000,Zurek2003}.
Comprehensive treatments of decoherence and its relation to the
measurement problem can be found in
Refs.~\cite{Schlosshauer2005,Schlosshauer2019}.

Gaussian pointer states are particularly useful in this setting
because conditional translations can be treated analytically. The
overlap between two translated Gaussian packets decreases with their
relative displacement, providing a natural measure of how effectively
the pointer distinguishes different spectral components of the
system. Once the pointer degrees of freedom are traced out, these
overlaps appear directly as suppression factors multiplying the
off-diagonal elements of the reduced density operator. Conditional
displacements also remain relevant in modern measurement and readout
protocols, including implementations in superconducting quantum
systems \cite{Touzard2019}. More recent studies have continued to
explore von Neumann measurement schemes with structured pointer states
\cite{YuanbekRenTurek2025}, while the characterization and emergence
of pointer states remain active topics in more general open-system
and measurement settings
\cite{KofmanKurizki2022,SinghSawickiKorbicz2024}.

A natural quantity for characterizing the reduced dynamics generated
by such correlations is the purity
$\mathcal{P}(t)=\operatorname{Tr}[\rho_S^2(t)]$.
For a globally pure bipartite state, a decrease in the purity of either
reduced subsystem directly reflects the generation of bipartite
entanglement between the two subsystems. More generally, purity and related
entropic quantities provide useful measures of reduced-state mixing
and entanglement generation. The short-time production of
entanglement has been studied in terms of fluctuations of the
interaction operators, leading to general entanglement timescales
\cite{ZnidariProsen2005,Yang2018,Cresswell2018}. Such results indicate
that the initial departure from purity contains information about the
variance structure of the interaction Hamiltonian.

Recent work has also emphasized the role of pairwise state overlaps in the
purity of density operators constructed from non-orthogonal quantum states.
In particular, overlap-based representations and bounds provide a direct
connection between the geometry of the underlying state set and the resulting
mixedness, with further implications for quantum Fisher information
\cite{Amirov2026Purity}. In the present setting, this overlap-based viewpoint
acquires a dynamical form: the relevant non-orthogonality arises between
Hamiltonian-conditioned pointer states, whose overlaps are determined
explicitly by the populated spectral separations of the system Hamiltonian.
This provides a natural route from pairwise state distinguishability to
spectral-gap-controlled purity dynamics.

For Hamiltonian-conditioned pointer interactions, however, it is also
important to distinguish the short-time behavior from the spectral
mechanisms governing the subsequent approach toward the asymptotic
reduced state. The exact reduced dynamics contains contributions from
all pairwise separations between the populated energy eigenvalues.
Consequently, the initial energy distribution, the smallest and
largest populated spectral gaps, the pointer width, and the
interaction strength play different roles in the purity dynamics.
While the standard conditional-translation picture explains the
suppression of individual coherences, it is useful to formulate
global estimates that connect the purity directly to extremal energy
separations on the spectral support actually populated by the initial
state.

In this work, we extend this overlap-based perspective to the present
dynamical setting. We consider a quantum system coupled to a
continuous-variable pointer through the interaction Hamiltonian
$\hat H_{\mathrm{int}}=g\,\hat H\otimes\hat p$,
where $\hat H$ is the system Hamiltonian and $\hat p$ is the pointer
momentum. For an initially Gaussian pointer state, we first derive the
conditional pointer translations and their exact overlaps. Tracing out the
pointer degrees of freedom then yields an explicit reduced density operator
and a closed analytical expression for the time-dependent purity in terms
of the populated pairwise energy separations.

Using this representation, we establish two-sided bounds on the purity
controlled by the minimum and maximum energy gaps on the populated
spectral support. The bounds are therefore state dependent: spectral
levels that are not populated initially do not enter the estimates.
We further show that these bounds are sharp for two-level systems,
where the minimum and maximum populated gaps coincide. The upper bound
is then used to obtain an explicit sufficient interaction timescale
for the purity to approach its asymptotic value within a prescribed
tolerance. In particular, the resulting timescale exhibits the scaling
$t_{\varepsilon}\propto\sigma/(|g|\Delta_{\min})$,
which identifies the minimum populated spectral gap as the slow
spectral scale governing this estimate.

The short-time regime provides a complementary characterization.
Expanding the exact purity around the initial time, we obtain
$\mathcal{P}''(0)=-(g^2/\sigma^2)
\operatorname{Var}_{\psi_S}(\hat H)$.
Thus, the initial curvature is governed by the energy variance of the
initial system state, whereas the sufficient long-time scale obtained
from the spectral bound is governed by the minimum populated energy
gap. This separates two different aspects of the spectral structure:
the weighted energy distribution controls the initial purity loss,
while the smallest relevant spectral separation determines the slowest
gap-dependent contribution entering the global upper bound.

Finally, we apply the general results to two representative classes
of states. The two-level case demonstrates the exact saturation of the
spectral-gap bounds, while an equally weighted $N$-level system with
an equally spaced spectrum provides an explicit many-level example.
For the latter, the asymptotic purity is $1/N$. The analytical and
numerical examples illustrate the dependence on the number of populated
levels and the spectral spacing, while the pointer-width dependence
follows directly from the equivalent $\Delta/\sigma$ scaling.

The main contribution of the present work is not the conditional
Gaussian translation itself, which is a standard feature of
von Neumann-type measurement interactions, but the spectral
characterization of the resulting purity dynamics. Specifically,
we derive two-sided purity bounds determined by the minimum and
maximum energy gaps on the populated spectral support, establish
their exact saturation for two-level systems, and obtain an explicit
sufficient timescale for approaching the asymptotic purity. We also
show that the short-time purity curvature is governed by the
Hamiltonian variance, thereby separating the spectral quantity that
controls the initial purity loss from the minimum populated gap that
controls the slowest contribution to the long-time upper bound.

The remainder of the paper is organized as follows. We first formulate
the Hamiltonian--pointer interaction and derive the associated
conditional translations. We then specialize to Gaussian pointer
states and analyze their spectral distinguishability. The reduced
system dynamics, exact purity, and spectral-gap bounds are subsequently
derived, followed by the short-time variance relation and the
gap-controlled purity timescale. We next establish the sharpness of
the bounds for two-level systems and examine the equally spaced
$N$-level model. The final sections discuss the physical implications,
limitations, and possible extensions of the results.

\section{Hamiltonian--Pointer Interaction}

Let $\mathcal{H}_S$ denote the Hilbert space of the quantum system
and let $\mathcal{H}_P=L^2(\mathbb{R})$ denote the Hilbert space
of a continuous-variable pointer.

Let $\hat{H}$ be the Hamiltonian of the quantum system. We assume
that $\hat{H}$ has the spectral decomposition
\begin{equation}
    \hat{H}
    =
    \sum_j E_j |E_j\rangle\langle E_j|,
    \label{eq:H_spectral}
\end{equation}
where $\hat{H}|E_j\rangle=E_j|E_j\rangle$.

The pointer is described by the canonical position and momentum
operators $\hat{x}$ and $\hat{p}$ satisfying
$[\hat{x},\hat{p}]=i\hbar$.

We consider a von Neumann-type interaction Hamiltonian of the form
\begin{equation}
    \hat H_{\mathrm{int}}
    =
    g\,\hat H\otimes\hat p,
    \label{eq:interaction_hamiltonian}
\end{equation}
where $g\in\mathbb{R}$ is the coupling constant. Since $\hat H$ has
the dimensions of energy and $\hat p$ has the dimensions of momentum,
$g$ has the dimensions of inverse momentum, equivalently
length divided by energy times time. Consequently,
$\hat H_{\mathrm{int}}$ has the dimensions of energy, while the
conditional displacement $gtE_j$ has the dimensions of length, as
required for a spatial translation of the pointer.

Whenever a nontrivial system--pointer interaction is considered,
we further assume that $g\neq0$.

For an interaction time $t$, the corresponding unitary evolution
operator is
\begin{equation}
    U(t)
    =
    \exp\left(
        -\frac{i}{\hbar}
        gt\,\hat{H}\otimes\hat{p}
    \right).
    \label{eq:unitary}
\end{equation}

Let the initial state of the quantum system be
\begin{equation}
    |\psi_S\rangle
    =
    \sum_j c_j |E_j\rangle,
    \qquad
    \sum_j |c_j|^2=1,
    \label{eq:system_state}
\end{equation}
and let the pointer initially be in a normalized state
$|\chi\rangle\in\mathcal{H}_P$.

Throughout the analysis, we assume that the initial system state has
finite spectral support. Whenever spectral-gap bounds and asymptotic
purity are considered, the populated energy eigenvalues are further
assumed to be pairwise distinct. This restriction is imposed only to
keep the extremal-gap definitions and the asymptotic purity statements
transparent. Degenerate spectra can be treated by grouping together
contributions associated with the same energy eigenvalue.

The initial state of the composite system is therefore
\begin{equation}
    |\Psi(0)\rangle
    =
    |\psi_S\rangle\otimes|\chi\rangle
    =
    \sum_j c_j|E_j\rangle\otimes|\chi\rangle.
    \label{eq:initial_state}
\end{equation}

\section{Conditional Translation of the Pointer}

We define the translation operator acting on the pointer Hilbert
space by
\begin{equation}
    T(d)
    =
    \exp\left(
        -\frac{i}{\hbar}d\hat{p}
    \right).
    \label{eq:translation_operator}
\end{equation}

Its action on a position eigenstate is
$T(d)|x\rangle=|x+d\rangle$. If
$\chi(x)=\langle x|\chi\rangle$ is the initial pointer wavefunction,
then $\langle x|T(d)|\chi\rangle=\chi(x-d)$.

\begin{theorem}[Hamiltonian-dependent conditional translation]
Let the initial composite state be given by
Eq.~\eqref{eq:initial_state}. Under the interaction Hamiltonian
in Eq.~\eqref{eq:interaction_hamiltonian}, the evolved state is
\begin{equation}
    |\Psi(t)\rangle
    =
    \sum_j
    c_j |E_j\rangle
    \otimes
    T(gtE_j)|\chi\rangle.
    \label{eq:main_result}
\end{equation}

Equivalently, defining
$|\chi_j(t)\rangle:=T(gtE_j)|\chi\rangle$, we obtain
\begin{equation}
    |\Psi(t)\rangle
    =
    \sum_j
    c_j |E_j\rangle
    \otimes
    |\chi_j(t)\rangle.
    \label{eq:pointer_decomposition}
\end{equation}

In the position representation, the conditional pointer states satisfy
\begin{equation}
    \chi_j(x,t)
    :=
    \langle x|\chi_j(t)\rangle
    =
    \chi(x-gtE_j).
    \label{eq:pointer_shift}
\end{equation}

Thus, each energy eigenvalue $E_j$ generates the pointer displacement
$d_j=gtE_j$.
\end{theorem}

\begin{proof}
Using the spectral decomposition of $\hat H$ in
Eq.~\eqref{eq:H_spectral}, the unitary operator can be written as
\begin{equation}
    U(t)
    =
    \sum_j
    |E_j\rangle\langle E_j|
    \otimes
    \exp\left(
        -\frac{i}{\hbar}gtE_j\hat p
    \right)
    =
    \sum_j
    |E_j\rangle\langle E_j|
    \otimes
    T(gtE_j).
    \label{eq:controlled_translation}
\end{equation}

Applying this operator to the initial composite state gives
\begin{align*}
    U(t)|\Psi(0)\rangle
    &=
    \left[
    \sum_j
    |E_j\rangle\langle E_j|
    \otimes
    T(gtE_j)
    \right]
    \left[
    \sum_k
    c_k |E_k\rangle\otimes|\chi\rangle
    \right]
    \\
    &=
    \sum_{j,k}
    c_k
    |E_j\rangle
    \langle E_j|E_k\rangle
    \otimes
    T(gtE_j)|\chi\rangle.
\end{align*}

Since the energy eigenstates are orthonormal,
$\langle E_j|E_k\rangle=\delta_{jk}$, it follows that
$|\Psi(t)\rangle=U(t)|\Psi(0)\rangle
=\sum_j c_j|E_j\rangle\otimes T(gtE_j)|\chi\rangle$.

With $|\chi_j(t)\rangle=T(gtE_j)|\chi\rangle$, the evolved state
can equivalently be expressed in terms of the conditional pointer
states introduced above.

Finally, using
$\langle x|T(d)|\chi\rangle=\chi(x-d)$ with $d=gtE_j$, we obtain
$\chi_j(x,t)=\chi(x-gtE_j)$. Hence, the energy eigenvalue $E_j$ is
encoded as a spatial displacement $gtE_j$ of the pointer state.
\end{proof}

\begin{remark}
The relations $T(d)|x\rangle=|x+d\rangle$ and
$\langle x|T(d)|\chi\rangle=\chi(x-d)$ are consistent. The first
describes the translation of a position eigenstate, whereas the second
describes the translated wavefunction in the position representation.
\end{remark}

\begin{remark}
If the coupling constant is absorbed into the interaction time,
or if dimensionless units are employed, one may formally set $g=1$.
In this case,
$|\Psi(t)\rangle=\sum_j c_j|E_j\rangle\otimes
T(tE_j)|\chi\rangle$, which corresponds directly to the original
formulation.
\end{remark}

\section{Gaussian Pointer States and Spectral Distinguishability}

We now specialize the initial pointer state to a normalized Gaussian
wave packet. This choice allows the conditional pointer overlaps to be
obtained analytically and reveals a direct connection between the
spectral gaps of the system Hamiltonian and the distinguishability of
the corresponding pointer states.

Let the initial pointer wavefunction be
\begin{equation}
    \chi(x)
    =
    \frac{1}{(2\pi\sigma^2)^{1/4}}
    \exp\left(
        -\frac{x^2}{4\sigma^2}
    \right),
    \label{eq:gaussian_pointer}
\end{equation}
where $\sigma>0$ denotes the spatial width of the pointer.

Following the interaction derived in the previous section, the
conditional pointer state associated with the energy eigenvalue $E_j$
is $|\chi_j(t)\rangle=T(gtE_j)|\chi\rangle$. In the position
representation, this state has the wavefunction
\begin{equation}
    \chi_j(x,t)
    =
    \frac{1}{(2\pi\sigma^2)^{1/4}}
    \exp\left[
        -\frac{(x-gtE_j)^2}{4\sigma^2}
    \right].
    \label{eq:translated_gaussian}
\end{equation}

\begin{theorem}[Overlap of Hamiltonian-conditioned pointer states]
Let $|\chi_j(t)\rangle$ and $|\chi_k(t)\rangle$ denote the
Gaussian pointer states associated with the energy eigenvalues
$E_j$ and $E_k$, respectively. Then their overlap is
\begin{equation}
    \langle\chi_k(t)|\chi_j(t)\rangle
    =
    \exp\left[
        -\frac{g^2t^2(E_j-E_k)^2}{8\sigma^2}
    \right].
    \label{eq:pointer_overlap}
\end{equation}
\end{theorem}

\begin{proof}
By definition,
$\langle\chi_k(t)|\chi_j(t)\rangle
=\int_{-\infty}^{\infty}
\chi_k^*(x,t)\chi_j(x,t)\,dx$.
Since the Gaussian wavefunctions in
Eq.~\eqref{eq:translated_gaussian} are real, this gives
\begin{align*}
    \langle\chi_k(t)|\chi_j(t)\rangle
    &=
    \frac{1}{\sqrt{2\pi\sigma^2}}
    \int_{-\infty}^{\infty}
    \exp\left[
        -\frac{(x-gtE_k)^2}{4\sigma^2}
        -\frac{(x-gtE_j)^2}{4\sigma^2}
    \right]dx.
\end{align*}

Introducing the displacements $d_j=gtE_j$ and $d_k=gtE_k$, and using
$(x-d_j)^2+(x-d_k)^2
=2\left(x-\frac{d_j+d_k}{2}\right)^2
+\frac{(d_j-d_k)^2}{2}$, we obtain
\begin{align*}
    \langle\chi_k(t)|\chi_j(t)\rangle
    &=
    \exp\left[
        -\frac{(d_j-d_k)^2}{8\sigma^2}
    \right]
    \frac{1}{\sqrt{2\pi\sigma^2}}
    \\
    &\quad\times
    \int_{-\infty}^{\infty}
    \exp\left[
        -\frac{1}{2\sigma^2}
        \left(
            x-\frac{d_j+d_k}{2}
        \right)^2
    \right]dx.
\end{align*}

The normalized Gaussian integral equals unity. Using
$d_j-d_k=gt(E_j-E_k)$ then directly yields
Eq.~\eqref{eq:pointer_overlap}.
\end{proof}

\begin{corollary}[Dependence on the spectral gap]
Define the spectral gap between the two energy eigenvalues by
$\Delta_{jk}:=|E_j-E_k|$. Then
$\left|\langle\chi_k(t)|\chi_j(t)\rangle\right|
=\exp[-g^2t^2\Delta_{jk}^2/(8\sigma^2)]$.
Consequently, for fixed $g\neq0$, $t>0$, and $\sigma>0$, the overlap is
a strictly decreasing function of $\Delta_{jk}$.
\end{corollary}

\begin{proof}
For fixed $g\neq0$, $t>0$, and $\sigma>0$, let
$f(\Delta)=\exp[-g^2t^2\Delta^2/(8\sigma^2)]$.
Differentiating with respect to $\Delta$ gives
$\frac{df}{d\Delta}
=-\frac{g^2t^2}{4\sigma^2}\Delta
\exp[-g^2t^2\Delta^2/(8\sigma^2)]<0$
for every $\Delta>0$. Hence, pointer states associated with larger
Hamiltonian spectral gaps have smaller overlap and are therefore more
distinguishable.
\end{proof}

\section{Reduced Dynamics and Purity}

The conditional translation of the pointer induces a nontrivial
reduced dynamics on the quantum system. Although the total
system--pointer state remains pure under the unitary evolution,
the reduced state of the system generally becomes mixed because
of the entanglement generated between the energy eigenstates and
the corresponding pointer states.

Let $p_j:=|c_j|^2$, with $\sum_j p_j=1$. For convenience, we also
introduce the energy gaps $\Delta_{jk}:=|E_j-E_k|$.

\begin{theorem}[Reduced system state]
Let the evolved system--pointer state be
\begin{equation}
    |\Psi(t)\rangle
    =
    \sum_j
    c_j |E_j\rangle
    \otimes
    |\chi_j(t)\rangle,
    \label{eq:joint_state_reduced}
\end{equation}
where the conditional pointer states are Gaussian and satisfy
\begin{equation}
    \langle\chi_k(t)|\chi_j(t)\rangle
    =
    \exp\left[
        -\frac{g^2t^2(E_j-E_k)^2}{8\sigma^2}
    \right].
    \label{eq:gaussian_overlap_reduced}
\end{equation}

Then the reduced density operator of the quantum system is
\begin{equation}
    \rho_S(t)
    =
    \sum_{j,k}
    c_j c_k^*
    \exp\left[
        -\frac{g^2t^2(E_j-E_k)^2}{8\sigma^2}
    \right]
    |E_j\rangle\langle E_k|.
    \label{eq:reduced_density_matrix}
\end{equation}
\end{theorem}

\begin{proof}
The density operator of the total system is
$\rho_{SP}(t)=|\Psi(t)\rangle\langle\Psi(t)|$. Using
Eq.~\eqref{eq:joint_state_reduced}, we obtain
\begin{align*}
    \rho_{SP}(t)
    &=
    \left(
        \sum_j
        c_j |E_j\rangle\otimes|\chi_j(t)\rangle
    \right)
    \left(
        \sum_k
        c_k^* \langle E_k|
        \otimes
        \langle\chi_k(t)|
    \right)
    \\
    &=
    \sum_{j,k}
    c_jc_k^*
    |E_j\rangle\langle E_k|
    \otimes
    |\chi_j(t)\rangle\langle\chi_k(t)|.
\end{align*}

The reduced density operator is obtained by tracing out the pointer
degrees of freedom,
$\rho_S(t)=\operatorname{Tr}_P[\rho_{SP}(t)]$. Since
$\operatorname{Tr}_P[
|\chi_j(t)\rangle\langle\chi_k(t)|
]
=
\langle\chi_k(t)|\chi_j(t)\rangle$,
we find
$\rho_S(t)
=
\sum_{j,k}
c_jc_k^*
\langle\chi_k(t)|\chi_j(t)\rangle
|E_j\rangle\langle E_k|$.

Substituting the Gaussian overlap from
Eq.~\eqref{eq:gaussian_overlap_reduced} directly yields
Eq.~\eqref{eq:reduced_density_matrix}, which proves the result.
\end{proof}

\begin{remark}[Energy-basis populations and coherences]
For $j=k$, Eq.~\eqref{eq:reduced_density_matrix} gives
$\rho_{jj}(t)=|c_j|^2=p_j$, so the energy-basis populations remain
unchanged.

For $j\neq k$,
\begin{equation}
    \rho_{jk}(t)
    =
    \rho_{jk}(0)
    \exp\left[
        -\frac{g^2t^2\Delta_{jk}^2}{8\sigma^2}
    \right].
    \label{eq:coherence_decay}
\end{equation}

Thus, coherences associated with larger spectral separations are
suppressed more rapidly by the system--pointer interaction.
\end{remark}

\begin{theorem}[Exact purity formula]
Under the assumptions of the previous theorem, the purity of the
reduced system,
$\mathcal{P}(t):=\operatorname{Tr}\left[\rho_S^2(t)\right]$,
is given exactly by
\begin{equation}
    \mathcal{P}(t)
    =
    \sum_{j,k}
    p_jp_k
    \exp\left[
        -\frac{g^2t^2(E_j-E_k)^2}{4\sigma^2}
    \right].
    \label{eq:exact_purity}
\end{equation}

Equivalently,
\begin{equation}
    \mathcal{P}(t)
    =
    \sum_j p_j^2
    +
    2\sum_{j<k}
    p_jp_k
    \exp\left[
        -\frac{g^2t^2\Delta_{jk}^2}{4\sigma^2}
    \right].
    \label{eq:exact_purity_pairwise}
\end{equation}
\end{theorem}

\begin{proof}
From Eq.~\eqref{eq:reduced_density_matrix}, define
$D_{jk}(t):=
\exp[-g^2t^2(E_j-E_k)^2/(8\sigma^2)]$.
Then
$\rho_S(t)
=
\sum_{j,k}
c_jc_k^*D_{jk}(t)
|E_j\rangle\langle E_k|$.

Squaring the reduced density operator gives
\begin{align*}
    \rho_S^2(t)
    &=
    \sum_{j,k,m,n}
    c_jc_k^*c_mc_n^*
    D_{jk}(t)D_{mn}(t)
    |E_j\rangle
    \langle E_k|E_m\rangle
    \langle E_n|.
\end{align*}

Using orthonormality,
$\langle E_k|E_m\rangle=\delta_{km}$, the summation over $m$
collapses. Hence,
$\rho_S^2(t)=\sum_{j,k,n} c_jc_k^*c_kc_n^*
D_{jk}(t)D_{kn}(t)|E_j\rangle\langle E_n|$.

Taking the trace imposes $n=j$, since
$\operatorname{Tr}(|E_j\rangle\langle E_n|)=\delta_{jn}$.
Consequently,
$\mathcal{P}(t)=\sum_{j,k}|c_j|^2|c_k|^2
D_{jk}(t)D_{kj}(t)$.

Since $D_{jk}(t)=D_{kj}(t)$, we have
$D_{jk}^2(t)=
\exp[-g^2t^2(E_j-E_k)^2/(4\sigma^2)]$.
Using $p_j=|c_j|^2$, it follows that
$\mathcal{P}(t)=\sum_{j,k}p_jp_k
\exp[-g^2t^2(E_j-E_k)^2/(4\sigma^2)]$.

Separating the diagonal and off-diagonal contributions yields
\begin{align*}
    \mathcal{P}(t)
    &=
    \sum_j p_j^2
    +
    \sum_{j\neq k}
    p_jp_k
    \exp\left[
        -\frac{g^2t^2\Delta_{jk}^2}{4\sigma^2}
    \right]
    \\
    &=
    \sum_j p_j^2
    +
    2\sum_{j<k}
    p_jp_k
    \exp\left[
        -\frac{g^2t^2\Delta_{jk}^2}{4\sigma^2}
    \right].
\end{align*}

This proves both forms of the exact purity formula.
\end{proof}

\begin{theorem}[Spectral-gap bounds on purity]
Assume that the initial system state has support on at least two
populated energy eigenstates with pairwise distinct energy eigenvalues.
Define the minimum and maximum energy separations on the populated
spectral support by
$\Delta_{\min}:=
\min_{\substack{j\neq k\\ p_jp_k>0}}|E_j-E_k|$
and
$\Delta_{\max}:=
\max_{\substack{j\neq k\\ p_jp_k>0}}|E_j-E_k|$.
Let $P_{\infty}:=\sum_j p_j^2$.

Then, for every $t\geq 0$,
\begin{equation}
\begin{aligned}
    P_{\infty}
    &+
    \left(1-P_{\infty}\right)
    \exp\left[
        -\frac{g^2t^2\Delta_{\max}^2}{4\sigma^2}
    \right]
    \\
    &\leq
    \mathcal{P}(t)
    \\
    &\leq
    P_{\infty}
    +
    \left(1-P_{\infty}\right)
    \exp\left[
        -\frac{g^2t^2\Delta_{\min}^2}{4\sigma^2}
    \right].
\end{aligned}
\label{eq:purity_bounds}
\end{equation}
\end{theorem}

\begin{proof}
From Eq.~\eqref{eq:exact_purity_pairwise},
$\mathcal{P}(t)
=
P_{\infty}
+
2\sum_{j<k}
p_jp_k
\exp[-g^2t^2\Delta_{jk}^2/(4\sigma^2)]$.

For every populated pair $j\neq k$,
$\Delta_{\min}\leq\Delta_{jk}\leq\Delta_{\max}$.
Since the exponential factor
$\exp[-g^2t^2\Delta^2/(4\sigma^2)]$ decreases with $\Delta$
for $t\geq0$, it follows that

\begin{align*}
    \exp\left[
        -\frac{g^2t^2\Delta_{\max}^2}{4\sigma^2}
    \right]
    &\leq
    \exp\left[
        -\frac{g^2t^2\Delta_{jk}^2}{4\sigma^2}
    \right]
    \\
    &\leq
    \exp\left[
        -\frac{g^2t^2\Delta_{\min}^2}{4\sigma^2}
    \right].
\end{align*}

Multiplying by the nonnegative quantities $2p_jp_k$ and summing
over all $j<k$ gives
\begin{align*}
    2
    \exp\left[
        -\frac{g^2t^2\Delta_{\max}^2}{4\sigma^2}
    \right]
    \sum_{j<k}p_jp_k
    &\leq
    2\sum_{j<k}
    p_jp_k
    \exp\left[
        -\frac{g^2t^2\Delta_{jk}^2}{4\sigma^2}
    \right]
    \\
    &\leq
    2
    \exp\left[
        -\frac{g^2t^2\Delta_{\min}^2}{4\sigma^2}
    \right]
    \sum_{j<k}p_jp_k.
\end{align*}

Using the normalization condition $\sum_jp_j=1$, we have
$1=(\sum_jp_j)^2
=\sum_jp_j^2+2\sum_{j<k}p_jp_k$.
Therefore,
$2\sum_{j<k}p_jp_k=1-P_{\infty}$.

Substitution into the preceding inequality gives
\begin{align*}
    P_{\infty}
    &+
    (1-P_{\infty})
    \exp\left[
        -\frac{g^2t^2\Delta_{\max}^2}{4\sigma^2}
    \right]
    \leq
    \mathcal{P}(t)
    \\
    &\leq
    P_{\infty}
    +
    (1-P_{\infty})
    \exp\left[
        -\frac{g^2t^2\Delta_{\min}^2}{4\sigma^2}
    \right].
\end{align*}

This completes the proof.
\end{proof}

\begin{corollary}[Initial and asymptotic purity]
Under the assumptions of the preceding theorem,
$\mathcal{P}(0)=1$.

If all populated energy eigenvalues are distinct and $g\neq0$, then
\begin{equation}
    \lim_{t\to\infty}\mathcal{P}(t)
    =
    \sum_j p_j^2
    =
    P_{\infty}.
    \label{eq:purity_long_time}
\end{equation}
\end{corollary}

\begin{proof}
Setting $t=0$ in Eq.~\eqref{eq:exact_purity} gives
$\mathcal{P}(0)
=\sum_{j,k}p_jp_k
=(\sum_jp_j)^2=1$.

For distinct populated energy eigenvalues,
$\Delta_{jk}>0$ whenever $j\neq k$. Therefore,
$\lim_{t\to\infty}
\exp[-g^2t^2\Delta_{jk}^2/(4\sigma^2)]=0$
for every $j\neq k$. Hence, only the diagonal terms remain, yielding
$\lim_{t\to\infty}\mathcal{P}(t)=\sum_jp_j^2=P_{\infty}$.
\end{proof}

\begin{corollary}[Monotonicity of the purity]
For $g\neq0$ and $t>0$, the purity satisfies
\begin{equation}
    \frac{d\mathcal{P}(t)}{dt}
    =
    -\frac{g^2t}{\sigma^2}
    \sum_{j<k}
    p_jp_k\Delta_{jk}^2
    \exp\left[
        -\frac{g^2t^2\Delta_{jk}^2}{4\sigma^2}
    \right]
    \leq 0.
    \label{eq:purity_derivative}
\end{equation}

If the initial state contains at least two populated components
with distinct energies, then
$d\mathcal{P}(t)/dt<0$ for every $t>0$.
\end{corollary}

\begin{proof}
Differentiating Eq.~\eqref{eq:exact_purity_pairwise} with respect
to $t$ gives
\begin{align*}
    \frac{d\mathcal{P}(t)}{dt}
    &=
    2\sum_{j<k}p_jp_k
    \left(
        -\frac{g^2t\Delta_{jk}^2}{2\sigma^2}
    \right)
    \exp\left[
        -\frac{g^2t^2\Delta_{jk}^2}{4\sigma^2}
    \right]
    \\
    &=
    -\frac{g^2t}{\sigma^2}
    \sum_{j<k}
    p_jp_k\Delta_{jk}^2
    \exp\left[
        -\frac{g^2t^2\Delta_{jk}^2}{4\sigma^2}
    \right].
\end{align*}

Every term in the sum is nonnegative, so
$d\mathcal{P}(t)/dt\leq0$.
If at least one populated pair satisfies $\Delta_{jk}>0$, the
corresponding term is strictly positive for $t>0$, and therefore
$d\mathcal{P}(t)/dt<0$.
\end{proof}

\section{Short-Time Purity Loss and Energy Variance}

The exact purity formula derived above also provides a direct
connection between the initial loss of purity and the energy
fluctuations of the quantum system. In particular, the short-time
behavior is governed by the variance of the system Hamiltonian in
the initial state.

Let
$\langle\hat{H}\rangle
:=\langle\psi_S|\hat{H}|\psi_S\rangle
=\sum_j p_jE_j$
and
$\langle\hat{H}^2\rangle
:=\langle\psi_S|\hat{H}^2|\psi_S\rangle
=\sum_j p_jE_j^2$.
The energy variance of the initial state is defined by
$\operatorname{Var}_{\psi_S}(\hat{H})
:=\langle\hat{H}^2\rangle-\langle\hat{H}\rangle^2$.

\begin{theorem}[Short-time purity loss and Hamiltonian variance]
For the Gaussian pointer interaction considered above, the purity
of the reduced system has the short-time expansion
\begin{equation}
    \mathcal{P}(t)
    =
    1
    -
    \frac{g^2t^2}{2\sigma^2}
    \operatorname{Var}_{\psi_S}(\hat{H})
    +
    O(t^4),
    \qquad
    t\rightarrow0.
    \label{eq:short_time_purity}
\end{equation}

Equivalently, $\mathcal{P}'(0)=0$, and
\begin{equation}
    \mathcal{P}''(0)
    =
    -\frac{g^2}{\sigma^2}
    \operatorname{Var}_{\psi_S}(\hat{H}).
    \label{eq:second_derivative_purity}
\end{equation}

Thus, the initial curvature of the purity is determined exactly by
the energy variance of the initial quantum state.
\end{theorem}

\begin{proof}
From Eq.~\eqref{eq:exact_purity},
$\mathcal{P}(t)
=\sum_{j,k}p_jp_k
\exp[-g^2t^2(E_j-E_k)^2/(4\sigma^2)]$.

For sufficiently small $t$, we use the Taylor expansion
$e^{-x}=1-x+O(x^2)$ as $x\rightarrow0$, with
$x=g^2t^2(E_j-E_k)^2/(4\sigma^2)$. Hence,
\[
    \exp\left[
        -\frac{g^2t^2(E_j-E_k)^2}{4\sigma^2}
    \right]
    =
    1
    -
    \frac{g^2t^2(E_j-E_k)^2}{4\sigma^2}
    +
    O(t^4).
\]

Substituting this expansion into the exact purity formula gives
\begin{align*}
    \mathcal{P}(t)
    &=
    \sum_{j,k}p_jp_k
    -
    \frac{g^2t^2}{4\sigma^2}
    \sum_{j,k}
    p_jp_k(E_j-E_k)^2
    +
    O(t^4).
\end{align*}

Because $\sum_jp_j=1$, we have
$\sum_{j,k}p_jp_k=(\sum_jp_j)^2=1$.
It therefore remains to evaluate
$S:=\sum_{j,k}p_jp_k(E_j-E_k)^2$.

Expanding the square gives
\begin{align*}
    S
    &=
    \sum_{j,k}
    p_jp_k
    \left(
        E_j^2+E_k^2-2E_jE_k
    \right)
    \\
    &=
    \sum_{j,k}p_jp_kE_j^2
    +
    \sum_{j,k}p_jp_kE_k^2
    -
    2\sum_{j,k}p_jp_kE_jE_k.
\end{align*}

The first term satisfies
$\sum_{j,k}p_jp_kE_j^2
=(\sum_jp_jE_j^2)(\sum_kp_k)
=\langle\hat{H}^2\rangle$.
Similarly,
$\sum_{j,k}p_jp_kE_k^2
=\langle\hat{H}^2\rangle$.

For the third term,
$\sum_{j,k}p_jp_kE_jE_k
=(\sum_jp_jE_j)(\sum_kp_kE_k)
=\langle\hat{H}\rangle^2$.
Therefore,
\[
    S
    =
    2\langle\hat{H}^2\rangle
    -
    2\langle\hat{H}\rangle^2
    =
    2\operatorname{Var}_{\psi_S}(\hat{H}).
\]

Substituting this identity into the short-time expansion gives
\begin{align*}
    \mathcal{P}(t)
    &=
    1
    -
    \frac{g^2t^2}{4\sigma^2}
    \left[
        2\operatorname{Var}_{\psi_S}(\hat{H})
    \right]
    +
    O(t^4)
    \\
    &=
    1
    -
    \frac{g^2t^2}{2\sigma^2}
    \operatorname{Var}_{\psi_S}(\hat{H})
    +
    O(t^4).
\end{align*}

This proves Eq.~\eqref{eq:short_time_purity}.
Differentiating the short-time expansion gives
$\mathcal{P}'(0)=0$ and
$\mathcal{P}''(0)
=-(g^2/\sigma^2)\operatorname{Var}_{\psi_S}(\hat{H})$,
in agreement with Eq.~\eqref{eq:second_derivative_purity}.

Hence, the initial curvature of the purity is completely
determined by the energy variance of the initial state.
\end{proof}

\begin{corollary}[Energy eigenstates are invariant in purity]
If the initial system state is an energy eigenstate
$|\psi_S\rangle=|E_m\rangle$, then
$\operatorname{Var}_{\psi_S}(\hat{H})=0$.
Moreover, the reduced system remains pure for all interaction times,
$\mathcal{P}(t)=1$ for every $t\geq0$.
\end{corollary}

\begin{proof}
If $|\psi_S\rangle=|E_m\rangle$, then $p_m=1$ and $p_j=0$
for every $j\neq m$. Therefore,
$\langle\hat{H}\rangle=E_m$ and
$\langle\hat{H}^2\rangle=E_m^2$, so
$\operatorname{Var}_{\psi_S}(\hat{H})=E_m^2-E_m^2=0$.

Furthermore, the joint state evolves as
$|\Psi(t)\rangle
=|E_m\rangle\otimes T(gtE_m)|\chi\rangle$.
Thus, the state remains separable for all $t$, and consequently
$\rho_S(t)=|E_m\rangle\langle E_m|$. It follows immediately that
$\mathcal{P}(t)=\operatorname{Tr}[\rho_S^2(t)]=1$.
\end{proof}

\section{Spectral-Gap-Controlled Purity Timescale}

The two-sided purity bounds obtained above can be used to derive
an explicit timescale for the approach of the reduced system
toward its asymptotic purity.

Recall that $P_{\infty}=\sum_j p_j^2$. From
Eq.~\eqref{eq:purity_bounds}, the upper spectral-gap bound is
\[
    \mathcal{P}(t)
    \leq
    P_{\infty}
    +
    \left(1-P_{\infty}\right)
    \exp\left[
        -\frac{g^2t^2\Delta_{\min}^2}{4\sigma^2}
    \right].
\]

We now quantify the time required for the purity to lie within a
prescribed tolerance of its asymptotic value.

\begin{theorem}[Spectral-gap-controlled purity timescale]
Assume that the initial system state has support on at least two
distinct energy eigenvalues, and let
$\Delta_{\min}
=\min_{\substack{j\neq k\\p_jp_k>0}}|E_j-E_k|>0$.
Let $\varepsilon$ satisfy
$0<\varepsilon<1-P_{\infty}$.

Define
\begin{equation}
    t_{\varepsilon}
    :=
    \frac{2\sigma}{|g|\Delta_{\min}}
    \sqrt{
        \ln\left(
            \frac{1-P_{\infty}}{\varepsilon}
        \right)
    }.
    \label{eq:purity_timescale}
\end{equation}

Then every interaction time satisfying $t\geq t_{\varepsilon}$
guarantees that
$0\leq\mathcal{P}(t)-P_{\infty}\leq\varepsilon$.
Thus, $t_{\varepsilon}$ provides a sufficient timescale for the
reduced-system purity to approach its asymptotic value within
the prescribed tolerance $\varepsilon$.
\end{theorem}

\begin{proof}
From Eq.~\eqref{eq:purity_bounds},
\[
    \mathcal{P}(t)-P_{\infty}
    \leq
    (1-P_{\infty})
    \exp\left[
        -\frac{g^2t^2\Delta_{\min}^2}{4\sigma^2}
    \right].
\]

Therefore, it is sufficient to require
\[
    (1-P_{\infty})
    \exp\left[
        -\frac{g^2t^2\Delta_{\min}^2}{4\sigma^2}
    \right]
    \leq
    \varepsilon.
\]

Since $0<\varepsilon<1-P_{\infty}$, dividing by
$1-P_{\infty}$ gives
\[
    \exp\left[
        -\frac{g^2t^2\Delta_{\min}^2}{4\sigma^2}
    \right]
    \leq
    \frac{\varepsilon}{1-P_{\infty}}.
\]

Taking the natural logarithm and multiplying by $-1$ yields
\[
    \frac{g^2t^2\Delta_{\min}^2}{4\sigma^2}
    \geq
    \ln\left(
        \frac{1-P_{\infty}}{\varepsilon}
    \right).
\]

Hence,
\[
    t^2
    \geq
    \frac{4\sigma^2}{g^2\Delta_{\min}^2}
    \ln\left(
        \frac{1-P_{\infty}}{\varepsilon}
    \right).
\]

Since $t\geq0$, it follows that
\[
    t
    \geq
    \frac{2\sigma}{|g|\Delta_{\min}}
    \sqrt{
        \ln\left(
            \frac{1-P_{\infty}}{\varepsilon}
        \right)
    }
    =
    t_{\varepsilon}.
\]

Therefore, every $t\geq t_{\varepsilon}$ satisfies
$\mathcal{P}(t)-P_{\infty}\leq\varepsilon$.

Furthermore, Eq.~\eqref{eq:exact_purity_pairwise} shows that
$\mathcal{P}(t)-P_{\infty}$ is a sum of nonnegative terms.
Consequently, $\mathcal{P}(t)\geq P_{\infty}$, and hence
$0\leq\mathcal{P}(t)-P_{\infty}\leq\varepsilon$.

This completes the proof.
\end{proof}

\begin{corollary}[Scaling of the purity timescale]
For fixed $P_{\infty}$ and fixed tolerance $\varepsilon$, the
sufficient purity timescale satisfies
$t_{\varepsilon}\propto\sigma/(|g|\Delta_{\min})$.

Consequently:
\begin{enumerate}
    \item increasing the minimum populated spectral gap
    $\Delta_{\min}$ decreases the purity timescale;

    \item increasing the coupling strength $|g|$ decreases the
    purity timescale;

    \item increasing the pointer width $\sigma$ increases the
    purity timescale.
\end{enumerate}
\end{corollary}

\begin{proof}
The result follows directly from Eq.~\eqref{eq:purity_timescale},
since
$\sqrt{\ln[(1-P_{\infty})/\varepsilon]}$
is constant when $P_{\infty}$ and $\varepsilon$ are fixed.
\end{proof}

\section{Sharpness for Two-Level Systems}

We now show that the spectral-gap purity bounds become exact for a
two-level system. Thus, the bounds derived above are sharp.

Consider the normalized initial state
$|\psi_S\rangle=c_1|E_1\rangle+c_2|E_2\rangle$,
where $p_1=|c_1|^2$, $p_2=|c_2|^2$, and $p_1+p_2=1$.
Define the unique nonzero energy gap by
$\Delta:=|E_1-E_2|$.

\begin{proposition}[Exact purity for a two-level system]
For the two-level system defined above, the reduced-system purity is
\begin{equation}
    \mathcal{P}(t)
    =
    p_1^2+p_2^2
    +
    2p_1p_2
    \exp\left[
        -\frac{g^2t^2\Delta^2}{4\sigma^2}
    \right].
    \label{eq:two_level_purity}
\end{equation}

Moreover, $\Delta_{\min}=\Delta_{\max}=\Delta$, and therefore both
the lower and upper bounds in Eq.~\eqref{eq:purity_bounds} coincide
with the exact purity. Hence, the spectral-gap bounds are sharp for
every two-level system with two populated energy eigenstates.
\end{proposition}

\begin{proof}
From Eq.~\eqref{eq:exact_purity}, the four contributions for a
two-level system are
\begin{align*}
    \mathcal{P}(t)
    &=
    p_1^2+p_2^2
    \\
    &\quad+
    p_1p_2
    \exp\left[
        -\frac{g^2t^2(E_1-E_2)^2}{4\sigma^2}
    \right]
    \\
    &\quad+
    p_2p_1
    \exp\left[
        -\frac{g^2t^2(E_2-E_1)^2}{4\sigma^2}
    \right].
\end{align*}

Since $(E_1-E_2)^2=(E_2-E_1)^2=\Delta^2$, this reduces to
Eq.~\eqref{eq:two_level_purity}.

Furthermore, because there is only one populated nonzero spectral
separation, $\Delta_{\min}=\Delta_{\max}=\Delta$.
Also, $P_{\infty}=p_1^2+p_2^2$, and using $p_1+p_2=1$, we have
$1-P_{\infty}=2p_1p_2$.

Therefore, both sides of the general purity bound reduce to
\[
    P_{\infty}
    +
    (1-P_{\infty})
    \exp\left[
        -\frac{g^2t^2\Delta^2}{4\sigma^2}
    \right],
\]
which is exactly Eq.~\eqref{eq:two_level_purity}. Hence, the bound
is saturated for all $t\geq0$.
\end{proof}

\begin{corollary}[Balanced two-level superposition]
For
$|\psi_S\rangle
=2^{-1/2}(|E_1\rangle+e^{i\varphi}|E_2\rangle)$,
we have $p_1=p_2=1/2$, and therefore
\begin{equation}
    \mathcal{P}(t)
    =
    \frac{1}{2}
    \left[
        1+
        \exp\left(
            -\frac{g^2t^2\Delta^2}{4\sigma^2}
        \right)
    \right].
    \label{eq:balanced_two_level_purity}
\end{equation}

In particular, $\mathcal{P}(0)=1$, whereas
$\lim_{t\rightarrow\infty}\mathcal{P}(t)=1/2$.
\end{corollary}

\begin{proof}
For the balanced state, $p_1=p_2=1/2$. Substituting these values
into Eq.~\eqref{eq:two_level_purity} gives
\begin{align*}
    \mathcal{P}(t)
    &=
    \frac{1}{4}
    +
    \frac{1}{4}
    +
    \frac{1}{2}
    \exp\left[
        -\frac{g^2t^2\Delta^2}{4\sigma^2}
    \right]
    \\
    &=
    \frac{1}{2}
    \left[
        1+
        \exp\left(
            -\frac{g^2t^2\Delta^2}{4\sigma^2}
        \right)
    \right].
\end{align*}

The limits follow immediately.
\end{proof}

\section{Equally Spaced $N$-Level Spectrum}

To illustrate the analytical results, we consider an $N$-level
system with equally spaced energy eigenvalues
$E_j=j\Delta$, where $j=0,1,\ldots,N-1$ and $\Delta>0$
denotes the fundamental level spacing.

We assume that the initial system state is the equal superposition
$|\psi_S\rangle=N^{-1/2}\sum_{j=0}^{N-1}|E_j\rangle$.
Hence, $p_j=1/N$ for all $j$.

\begin{proposition}[Exact purity for an equally spaced spectrum]
For the equally spaced $N$-level system defined above, the
reduced-system purity is
\begin{equation}
    \mathcal{P}(t)
    =
    \frac{1}{N}
    +
    \frac{2}{N^2}
    \sum_{m=1}^{N-1}
    (N-m)
    \exp\left[
        -\frac{g^2t^2\Delta^2m^2}{4\sigma^2}
    \right].
    \label{eq:N_level_exact_purity}
\end{equation}
\end{proposition}

\begin{proof}
Starting from Eq.~\eqref{eq:exact_purity} and using
$p_j=1/N$ and $E_j-E_k=(j-k)\Delta$, we obtain
\[
    \mathcal{P}(t)
    =
    \frac{1}{N^2}
    \sum_{j,k=0}^{N-1}
    \exp\left[
        -\frac{g^2t^2\Delta^2(j-k)^2}{4\sigma^2}
    \right].
\]

The diagonal contribution $j=k$ consists of $N$ terms and therefore
equals $1/N$. For a fixed positive integer $m=|j-k|$, there are
$2(N-m)$ ordered pairs $(j,k)$ satisfying $|j-k|=m$. Hence, the
off-diagonal contribution is
\[
    \frac{2}{N^2}
    \sum_{m=1}^{N-1}
    (N-m)
    \exp\left[
        -\frac{g^2t^2\Delta^2m^2}{4\sigma^2}
    \right].
\]

Combining the diagonal and off-diagonal contributions yields
Eq.~\eqref{eq:N_level_exact_purity}.
\end{proof}

\begin{corollary}[Asymptotic purity]
For the equally weighted $N$-level state,
\begin{equation}
    \lim_{t\to\infty}\mathcal{P}(t)
    =
    \frac{1}{N}.
    \label{eq:N_level_asymptotic_purity}
\end{equation}
\end{corollary}

\begin{proof}
Since $\Delta>0$, every term with $m\geq1$ in
Eq.~\eqref{eq:N_level_exact_purity} vanishes as $t\to\infty$.
Hence, only the diagonal contribution remains, giving
$\lim_{t\to\infty}\mathcal{P}(t)=1/N$.
\end{proof}

\begin{figure}[!htbp]
    \centering
    \includegraphics[width=0.72\linewidth]{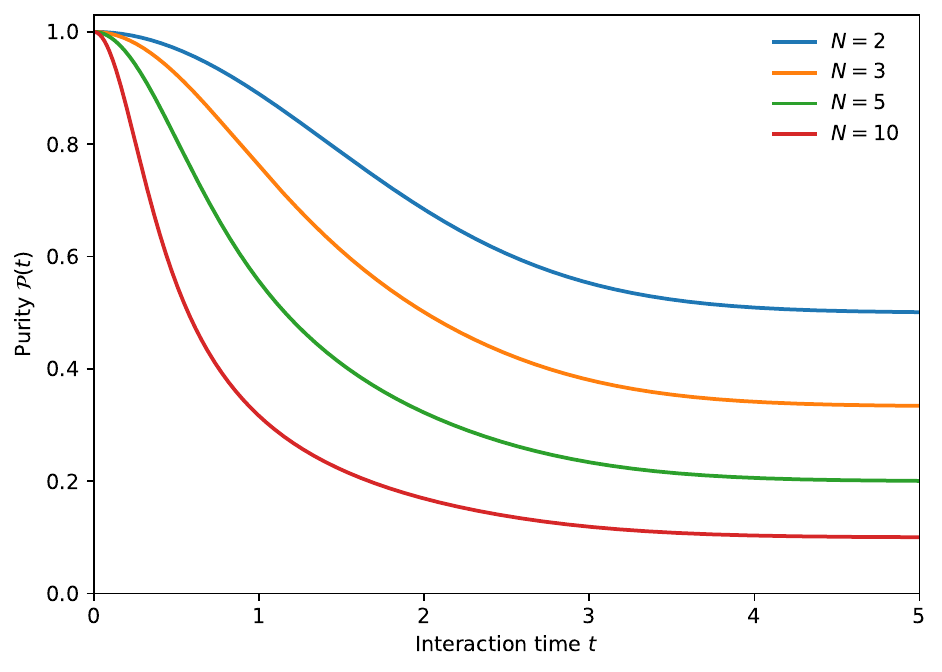}
    \caption{
    Time evolution of the reduced-system purity for equally weighted
    systems with $N=2,3,5,$ and $10$ equally spaced energy levels.
    The parameters are fixed at $g=1$, $\Delta=1$, and $\sigma=1$.
    The purity starts from $\mathcal{P}(0)=1$ and approaches the
    asymptotic value $\mathcal{P}(\infty)=1/N$.
    }
    \label{fig:purity_N}
\end{figure}

Figure~\ref{fig:purity_N} illustrates the dimensional dependence
predicted by Eq.~\eqref{eq:N_level_asymptotic_purity}. Although all
curves originate from the pure initial value $\mathcal{P}(0)=1$,
their long-time limits depend explicitly on the number of populated
energy levels. In particular, increasing $N$ lowers the asymptotic
purity to $1/N$, reflecting the increasing number of populated
energy components whose associated pointer states become
asymptotically distinguishable.

\begin{corollary}[Purity bounds for an equally spaced spectrum]
For the equally spaced $N$-level system,
$\Delta_{\min}=\Delta$, $\Delta_{\max}=(N-1)\Delta$, and
$P_{\infty}=1/N$. Therefore,
\begin{equation}
\begin{aligned}
    \frac{1}{N}
    &+
    \left(1-\frac{1}{N}\right)
    \exp\left[
        -\frac{g^2t^2(N-1)^2\Delta^2}{4\sigma^2}
    \right]
    \\
    &\leq
    \mathcal{P}(t)
    \\
    &\leq
    \frac{1}{N}
    +
    \left(1-\frac{1}{N}\right)
    \exp\left[
        -\frac{g^2t^2\Delta^2}{4\sigma^2}
    \right].
\end{aligned}
\label{eq:N_level_bounds}
\end{equation}
\end{corollary}

\begin{figure}[!htbp]
    \centering
    \includegraphics[width=0.72\linewidth]{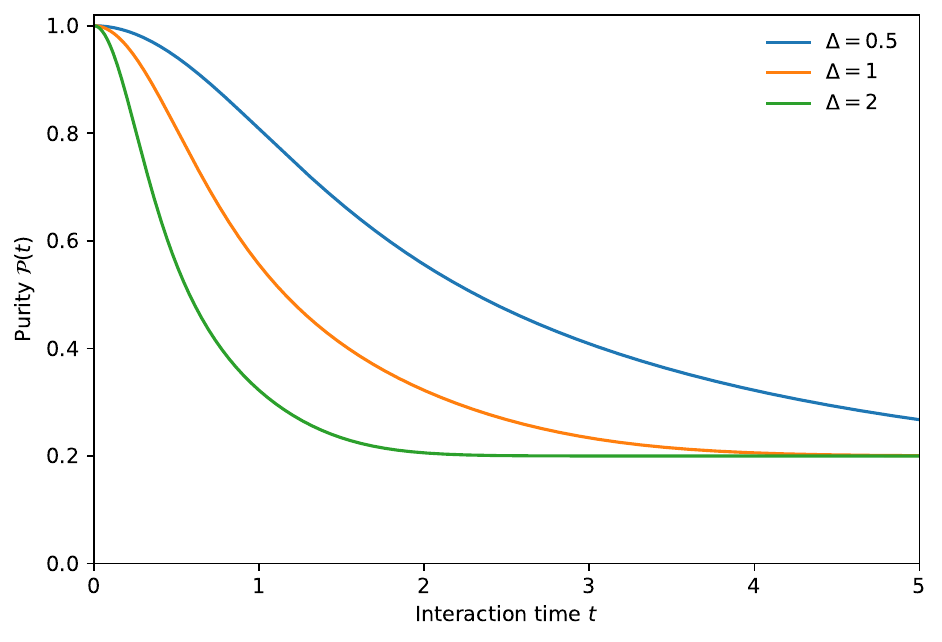}
    \caption{
    Dependence of the reduced-system purity on the fundamental
    spectral spacing $\Delta$ for an equally weighted five-level
    system with $g=1$ and $\sigma=1$. Increasing the spectral
    spacing accelerates the suppression of the off-diagonal
    coherences and leads to a faster approach toward the
    asymptotic purity $\mathcal{P}(\infty)=1/N$.
    }
    \label{fig:purity_gap}
\end{figure}

Figure~\ref{fig:purity_gap} confirms the spectral-gap dependence
predicted analytically. Larger values of $\Delta$ produce more
rapidly distinguishable conditional pointer states and hence
accelerate the decay of the off-diagonal contributions to the
purity. This behavior is consistent with both the exact expression
and the spectral-gap-controlled timescale derived above.

\subsection{Influence of the Pointer Width}

The analytical timescale derived above predicts the scaling
$t_{\varepsilon}\propto\sigma/(|g|\Delta_{\min})$.
Therefore, the width of the initial pointer wave packet provides a
direct control parameter for the rate of purity loss.

For the equally spaced model considered above, however, the dependence
on the pointer width is not independent of the spectral-spacing
dependence. Indeed, the exact purity can be written as
\begin{equation}
    \mathcal{P}(t)
    =
    \frac{1}{N}
    +
    \frac{2}{N^2}
    \sum_{m=1}^{N-1}
    (N-m)
    \exp\left[
        -\frac{g^2t^2m^2}{4}
        \left(\frac{\Delta}{\sigma}\right)^2
    \right].
    \label{eq:purity_gap_width_ratio}
\end{equation}

Thus, for fixed $N$ and $g$, the effects of the spectral spacing
$\Delta$ and the pointer width $\sigma$ enter the purity dynamics
through the ratio $\Delta/\sigma$. Consequently, increasing
$\Delta$ has the same dynamical effect as decreasing $\sigma$ by the
corresponding factor.

This observation explains why a separate pointer-width plot would
reproduce the same family of purity curves already displayed for
different spectral spacings, up to a relabeling of the parameter
values. A narrower pointer produces more rapidly distinguishable
conditional pointer states and therefore accelerates the loss of
purity, whereas a broader pointer increases their overlap and slows
the purity reduction.

\section{Discussion}
\label{sec:discussion}

The results obtained in this work provide a direct connection between
the spectral structure of the system Hamiltonian and the purity loss
generated by its interaction with a continuous-variable pointer.
Although the underlying Hamiltonian--pointer coupling produces the
standard conditional translation of the pointer wave packet, tracing
out the pointer leads to a reduced dynamics whose purity can be
characterized explicitly in terms of the populated energy gaps.

For an initially Gaussian pointer, the overlap between two
conditionally translated pointer states associated with energies
$E_j$ and $E_k$ is governed by
Eq.~\eqref{eq:pointer_overlap}. Consequently, coherences between
energy components with larger spectral separation are suppressed more
rapidly. This establishes a simple physical interpretation of the
reduced dynamics: the loss of coherence is controlled by the
distinguishability of the corresponding pointer states. Large energy
separations generate larger conditional translations and therefore
reduce the overlap between the associated pointer wave packets.

At the level of purity, this mechanism becomes particularly
transparent. Equation~\eqref{eq:exact_purity_pairwise} shows that the
purity is a weighted sum of gap-dependent decay factors, with
$P_{\infty}=\sum_jp_j^2$ representing the asymptotic contribution
associated with the populations. Thus, the long-time purity depends
on the population distribution over the energy eigenstates, whereas
the rate at which this limit is approached depends on the spectral
separations between the populated levels.

A central consequence of this representation is the two-sided
spectral-gap estimate in Eq.~\eqref{eq:purity_bounds}. Importantly,
$\Delta_{\min}$ and $\Delta_{\max}$ are defined on the populated
spectral support rather than on the entire spectrum of the
Hamiltonian. Energy levels that are absent from the initial state
therefore do not influence the resulting purity estimates. This
makes the bounds explicitly state dependent and avoids introducing
spectral scales that are dynamically irrelevant for the chosen
initial condition.

The roles of the two extremal gaps are complementary. The largest
populated separation $\Delta_{\max}$ controls the lower bound,
whereas the smallest populated nonzero separation $\Delta_{\min}$
controls the upper bound and hence the slowest contribution to the
approach toward $P_{\infty}$. In this sense, the smallest populated
spectral gap acts as a bottleneck for the long-time suppression of
the remaining coherence. This observation also explains the
appearance of $\Delta_{\min}$ in the sufficient purity timescale
derived in this work.

Indeed, Eq.~\eqref{eq:purity_timescale} makes the relevant physical
scales explicit. Stronger coupling and larger minimum populated
energy separation reduce the sufficient interaction time, whereas
a broader pointer wave packet increases it. The latter behavior
follows naturally from the fact that two translated Gaussian packets
with a fixed separation have a larger overlap when their spatial
width is increased. The numerical behavior shown above is consistent
with these analytical scaling relations.

The short-time regime reveals a different but complementary
connection between purity and the energy distribution. In particular,
Eqs.~\eqref{eq:short_time_purity} and
\eqref{eq:second_derivative_purity} show that the energy variance
determines the initial curvature of the purity. Hence, two different
spectral quantities characterize two distinct regimes of the
evolution. The energy variance determines the initial purity loss,
whereas the minimum populated spectral gap determines a sufficient
long-time scale for approaching the asymptotic value. The former
contains information about the entire weighted energy distribution,
whereas the latter identifies the slowest spectral separation
relevant to the initial state.

The energy-eigenstate limit provides an instructive special case.
If the system initially occupies a single energy eigenstate, its
energy variance vanishes and no entanglement between distinct energy
sectors can be generated. The pointer is translated as a whole,
while the system remains in the same pure energy state. Accordingly,
$\mathcal{P}(t)=1$ for all interaction times. This illustrates that
the purity reduction considered here originates from the coexistence
of distinct populated energy components rather than from the pointer
translation itself.

The two-level case further demonstrates that the general
spectral-gap bounds are not merely qualitative estimates. When two
distinct energy eigenstates are populated, there is only one nonzero
spectral separation, so
$\Delta_{\min}=\Delta_{\max}=\Delta$. The lower and upper bounds then
coincide with the exact purity for all interaction times. Thus, the
general bounds are saturated in the two-level case, establishing
their sharpness within this class of states.

The equally spaced $N$-level example illustrates how these results
extend beyond two-level systems. Equation~\eqref{eq:N_level_exact_purity}
provides the corresponding exact purity, with
$\mathcal{P}(\infty)=1/N$. Increasing the number of equally populated
energy levels therefore reduces the asymptotic purity, while
increasing the fundamental spacing $\Delta$ accelerates the approach
toward this limiting value. Together with the pointer-width
dependence, these examples illustrate separately the roles played by
the spectral dimension, spectral separation, and pointer resolution.

The present analysis is nevertheless based on several simplifying
assumptions. The pointer is initialized in a Gaussian pure state,
the interaction is generated solely by the bilinear coupling
$g\hat H\otimes\hat p$, and additional free evolution or
environmental noise is not included. Moreover, the asymptotic
expressions assume distinct populated energy eigenvalues; spectral
degeneracies require grouping contributions belonging to equal
energies. Extensions to non-Gaussian or mixed pointer states,
time-dependent coupling strengths, degenerate spectra, and open
system dynamics may therefore lead to more general forms of the
purity bounds and constitute natural directions for further study.

\section{Conclusion}
\label{sec:conclusion}

In this work, we investigated the purity dynamics of a quantum
system coupled to a continuous-variable pointer through the
interaction Hamiltonian
$\hat H_{\mathrm{int}}=g\,\hat H\otimes\hat p$.
For an initially Gaussian pointer, the conditional translations
associated with different system energies lead to an explicit
gap-dependent suppression of the reduced-system coherences.

Using this structure, we obtained an exact expression for the
time-dependent purity and derived two-sided bounds determined by
the minimum and maximum energy separations on the populated
spectral support. These bounds provide a direct connection between
the spectral geometry of the initial state and the loss of purity.
They are shown to be sharp for two-level systems, for which the
minimum and maximum populated gaps coincide.

The short-time analysis further showed that the initial curvature
of the purity is governed exactly by the Hamiltonian variance,
$\mathcal{P}''(0)
=-(g^2/\sigma^2)\operatorname{Var}_{\psi_S}(\hat H)$,
while the upper spectral-gap bound yields an explicit sufficient
timescale for approaching the asymptotic purity within a prescribed
tolerance. In particular, this timescale scales as
$t_{\varepsilon}\propto\sigma/(|g|\Delta_{\min})$, revealing the
competing roles of pointer width, interaction strength, and the
minimum populated spectral gap.

Finally, the equally spaced $N$-level model provided an explicit
many-level realization of the general results. For an equally
weighted initial state, the asymptotic purity becomes $1/N$.
The analytical and numerical examples illustrate the dependence on
the number of populated levels and the spectral spacing, whereas the
pointer-width dependence follows directly from the equivalent
$\Delta/\sigma$ scaling.

These results provide a compact spectral characterization of purity
loss in Hamiltonian-conditioned pointer interactions. Extensions to
non-Gaussian and mixed pointer states, degenerate spectra,
time-dependent interactions, and noisy open-system settings may
provide useful directions for further investigation.

\section*{Declaration of AI Usage}

During the preparation of this manuscript, the authors utilized an
AI-based language model solely for the purposes of language refinement,
grammar correction, and assistance with LaTeX formatting. The AI was
not used to generate, modify, or interpret any scientific content,
derivations, results, or conclusions presented in this work. All
theoretical developments, analytical derivations, numerical
calculations, and scientific interpretations were carried out
independently by the authors. The authors take full responsibility
for the accuracy and integrity of the content presented in this
manuscript.

\end{document}